\documentclass[letterpaper]{IEEEtran}

\usepackage{graphics} 
\usepackage{amsmath} 
\usepackage{amsthm}

\usepackage{amssymb}  
\usepackage{color}
\usepackage{multirow}
\usepackage{bm}
\usepackage{cite}
\usepackage{amsmath,amssymb}
\usepackage{algorithmic}
\usepackage{graphicx}
\usepackage{graphics} 
\usepackage{epsfig} 
\usepackage{amsmath} 
\usepackage{amssymb}  
\usepackage{algorithm}

\usepackage{algorithmic}

\usepackage{graphicx}           	
\usepackage{epstopdf}           	
\usepackage{subfloat}				
\usepackage{subfigure}
\usepackage{url}
\usepackage{tikz}
\usepackage{verbatim}
\usepackage{xcolor}
\usepackage{soul,color}

\newcommand{\V}{\mathcal{V}}
\newcommand{\E}{\mathcal{E}}

\newtheorem{definition}{Definition}

\newtheorem{conjecture}{Conjecture}

\newtheorem{assumption}{Assumption}
\newtheorem{proposition}{Proposition}
\newtheorem{remark}{Remark}

\newcommand{\vx}{\boldsymbol{x}}

\newcommand{\vy}{\boldsymbol{y}}

\renewcommand{\eqref}[1]{Eq.~(\ref{#1})}

\newcommand{\real}{\mathbb{R}}

\newcommand{\integernonnegative}{\mathbb{Z}_{\ge 0}}

\newcommand{\Exp}{\mathds{E}} 
\def\Exp{\mathbb{E}}

\title{\LARGE \bf
Title
}

	\title{Optimal policy design for decision problems under social influence}
	\author{Valentina Breschi,  Chiara Ravazzi,  Paolo Frasca,   Fabrizio Dabbene, and Mara Tanelli	\thanks{V. Breschi is with the Department of Electrical Engineering of Eindhoven University of Technology, P.O. Box 513, 5600 MB, Eindhoven, The Netherlands. (e-mail: v.breschi@tue.nl). \newline\indent
 M. Tanelli is with the Department of Electronics, Information and Bioengineering of the Politecnico di Milano, Piazza L. da Vinci, 32, 20133 Milano, Italy. (e-mail: mara.tanelli@polimi.it). \newline\indent C. Ravazzi and F. Dabbene are with the Institute of Electronics, Computer and Telecommunication Engineering, National Research Council of Italy (CNR-IEIIT), c/o Politecnico di Torino, Corso Duca Degli Abruzzi, 10129, Torino, Italy. (e-mail: name.surname@ieiit.cnr.it).	\newline\indent	P. Frasca is with Univ.\ Grenoble Alpes, CNRS, Inria, Grenoble INP, GIPSA-lab, F-38000 Grenoble, France and is Research Associate at the IEIIT-CNR, Torino, Italy. E-mail:paolo.frasca@gipsa-lab.fr 
 \newline \indent
 This work is partially supported by PRIN project TECHIE: “A control and network-based approach for fostering the adoption of new technologies in the ecological transition” Cod. 2022KPHA24 CUP: D53D23001320006.
  }}

\begin{document}

\maketitle
\thispagestyle{empty}
\pagestyle{empty}

\begin{abstract}
This paper focuses on describing the impact of policy actions on individuals' opinions in the presence of social and external influences toward proposing preliminary nudging strategies to achieve a cost-effectiveness trade-off. To this end, we extend the classical Friedkin and Johnsen model of opinion dynamics to incorporate random factors, such as variability in individual predispositions due to uncontrolled events (e.g., modeling the impact of the weather on daily mobility choices), and describe the impact of personalized policies. Furthermore, we formulate an optimal control problem aimed at fostering the social acceptance of particular actions/choices within the network. Through our analysis and numerical simulations, we illustrate the features of the proposed model in the absence of nudging and the effectiveness of the proposed (optimal) nudging strategies.
\end{abstract}
\section{Introduction}
\label{sec:intro}
Individuals frequently make repeated choices that affect their daily life, e.g. selecting a mode of transport to commute or deciding whether or not to use a certain service. These choices are generally influenced by $(i)$ the social environment individuals are immersed in \cite{friedkin1986formal} and $(ii)$ a resistance (a \emph{stubbornness}) to change, deriving from a limited predisposition toward specific actions due to rational or psychological factors \cite{friedkin1990social}. A clear example of such resistance can be found in the adoption of new technologies, such as, e.g., solar panels, that often encounter doubt or resistance from potential users, despite their positive impacts on society and the environment in the long run. 

At the same time, stakeholders and policymakers can enact tailor-made policies to encourage specific virtuous actions, e.g., facilitate a widespread adoption of green mobility services. To be successful, these  policies should consider the complex relationships between individuals, their characteristics and connections, and the consequences of past policy decisions \cite{breschi2022driving,breschi2022fostering}. While the works in \cite{breschi2022driving,breschi2022fostering} focus on designing and investigating the impact of nudging policies in a deterministic setting, some recent works have started investigating the influence of recommending systems (and, thus, slightly different \textquotedblleft policymakers\textquotedblright \ with respect to the ones considered in \cite{breschi2022driving,breschi2022fostering}) on the dynamics within social networks. Building on\footnote{It introduced a feedback interaction between a recommending system and a single user.} \cite{rossi2022closed}, \cite{castro2018opinion,goyal2019maintaining} have integrated the dynamics of social interactions and recommending system within influence networks. Meanwhile, the paper \cite{sprenger2024control} introduces a model-based control approach for user engagement maximization, based on a Model Predictive Control (MPC) strategy that takes into account the evolution of opinions.

\textbf{Goal.} In this paper, our goal is to develop new strategies to effectively leverage control mechanisms within repetitive decision-making scenarios to nudge individuals toward a target choice (e.g., promoting the use of bike sharing).

\textbf{Contribution.}  To achieve the previous goal, we initially propose a new model based on the classical Friedkin and Johnsen model of opinion dynamics \cite{friedkin1990social}. While leveraging the ability of this well-established opinion dynamics model to capture the intricate interplay between individual behaviors and social influences, we extend it by incorporating two additional external factors. First, we include the impact of tailor-made policies enacted (and controlled) by policymakers and/or stakeholders. Compared to \cite{sprenger2024control}, the policymaker is thus not intended as a new node introduced in the network, able to influence at each time individuals with an equal action, but is modeled as a personalized external influence.  Second, we include random factors that allow us to characterize phenomena such as the variability of predisposition and the random nature of social acceptance towards specific actions a policymaker seeks to promote. In the absence of external control action, the resulting individuals' opinions do not converge but persistently oscillate over time, as exemplified in a simple numerical example. Nonetheless, under suitable assumptions, we prove that the stochastic process representing the evolution of the individuals' opinion over time converges almost surely to a final limit and is ergodic. Therefore, over a sufficiently long period, the time averages of individuals' inclinations converge to their expected values.

Based on the proposed extension of the Friedkin and Johnsen model, we then formulate an optimal control problem for the design of personalized policies to foster the propagation and acceptance of a target choice within the network. When proposing our policy design strategy, unlike prior literature, we focus on the (realistic) case in which the policymaker is restricted to observing only the realizations of individuals' opinions. This choice leads us to propose two alternative approximations of the initial optimal control problem, resulting in quadratic losses that are standard in optimal and model predictive control (MPC) frameworks. The performance of all the proposed strategies is extensively analyzed over a (simple) numerical case study, showcasing their respective downsides and benefits. 

\textbf{Organization.} The paper is structured as follows. Section~\ref{sec:model} introduces the proposed model for opinion dynamics along with its open-loop properties. Section~\ref{sec:control} shifts to a closed-loop setting, introducing a suite of approximated policy design strategies intended to nudge individuals toward positive biases with respect to a target choice to be made. The impact of these policies is analyzed in Section~\ref{sec:numerical} by a numerical example. The paper is ended with some remarks and directions for future work. 

\textbf{Notation and useful definitions.} The set of positive integers (including zero) is denoted as $\mathbb{Z}_{\geq 0}$. Given a set $\mathcal{A} \subseteq \mathbb{R}^{n}$, its cardinality is indicated as $|\cal{A}|$. Vectors and matrices are indicated with bold symbols, i.e., $\boldsymbol{a} \in \mathbb{R}^{n}$ and $\boldsymbol{A} \in \mathbb{R}^{n \times m}$. The $i$-th element of a vector $\boldsymbol{a} \in \mathbb{R}^{n}$ is indicated as $a_{i}$. A diagonal matrix with elements equal (and equally ordered) to those in $\boldsymbol{A}$ is denoted as $[\boldsymbol{A}]$. Given a random vector $\boldsymbol{a} \in \mathbb{R}^{n}$, its expected value is denoted as $\mathbb{E}[\boldsymbol{a}]=\bar{\boldsymbol{a}}$. A random process $\boldsymbol{a}(t)$, for $t\in \mathbb{Z}_{\geq 0}$, is said to be ergodic if it satisfies the following definition.
\begin{definition}[Ergodic process]\label{def:ergodicity}
Given a random process $\{\boldsymbol{a}(t)\}_{t \in \mathbb{Z}_{\geq 0}}$, if there exists a random variable $\boldsymbol a_{\infty} \in \real^n $ such that 
\begin{equation}\label{eq:ergodicity}
\lim_{t\to \infty}\frac{1}{t}\sum_{\ell=0}^{t-1}\boldsymbol a({\ell})=\Exp[\boldsymbol a_{\infty}],
\end{equation} 
almost surely, then the process $\{\boldsymbol a(t)\}_{t\in\integernonnegative}$ will be said {\em ergodic}.
\end{definition}
The time average in~\eqref{eq:ergodicity} is called Ces\'aro average.

\section{The Model \& its free evolution}
\label{sec:model}
Consider a social network characterized by a directed graph $ \mathcal{G} =(\mathcal{V},\mathcal{E},\boldsymbol{P})$, whose nodes $v\in \mathcal{V}$ represent the individuals (or agents) in the network, the edges $\mathcal{E}$ indicate if they interact or not\footnote{If $(v,w)\in\E$ we say that $v$ and $w$ influence each others.}, while interpersonal influences are encoded in a weight matrix $\boldsymbol P$. Such a matrix is adapted to the graph, i.e.,
\begin{subequations}
\begin{equation}
    P_{v,w}\neq 0\iff (v,w)\in\E,
\end{equation}
and is here assumed to be stochastic and nonnegative, namely
\begin{equation}
    \sum_{w\in\V}P_{vw}=1,~~\forall v\in\cal{V}.
\end{equation}
Through this matrix, we can then denote the set of agents followed by $v \in \cal{V}$ as
\begin{equation}
    \mathcal{N}_v=\{w\in\V: P_{vw}>0\}.
\end{equation}
\end{subequations}

Each agent in $\mathcal{G}$ is endowed with a state $x_v(t)\in [0,1]$, $\forall v\in\V$, that represents the latent \emph{opinion/belief} at each time $t \in \mathbb{Z}_{\geq 0}$ about a specific (target) decision, e.g., renting a shared bike to commute. Specifically, a value $x_v(t)$ close to $0$ indicates that the $v$-th agent is not inclined to make such a target choice, while the closer the opinion to $1$, the more the individual's attitude is positive to the target decision. 

Similarly to \cite{friedkin1990social}, here we propose to model the evolution of the agents' opinion according to $(i)$ interactions with the neighbors according to the features of $\cal{V}$, i.e., \emph{social influence}, and $(ii)$ the individual biases $u_{\mathrm{o},v} \in [0,1]$, with $v \in \cal{V}$. Nonetheless, differently from \cite{friedkin1990social}, we additionally model the impact of \emph{external factors} that can shape individual choices. In particular, we model the impact of $(iii)$ \emph{controlled policies} $u_{v}^{\mathrm{c}}(t) \in \mathcal{U}_{v} \subset [0,1]$, for $v \in \cal{V}$, possibly tailor-made by a policy maker to meet individual's needs and nudge a change in the individual inclination (e.g., through incentive strategies, information campaigns), and $(iv)$ \emph{uncontrollable disturbances} $u_{v}^{\mathrm{nc}}(t) \in \Phi \subseteq [0,1]$, with $v \in \mathcal{V}$, representing the effect of possible (yet unpredictable) on individual opinions (e.g., changes in the weather that might influence one's choice of taking a shared bike). Due to the unpredictability of this last component, throughout the paper, we assume the following.
\begin{assumption}[Unpredictable external inputs]\label{ass:external_factors}
The external inputs $u_{v}^{\mathrm{nc}}(t)$, with $t \in \mathbb{Z}_{\geq 0}$, are realizations of a \emph{zero-mean, white} sequence, uniformly distributed in the interval $\delta[-u_{\mathrm{o},v},u_{\mathrm{o},v}]$, with $0<\delta\ll 1$, for all $v \in\cal{V}$.   
\end{assumption}
Note that Assumption~\ref{ass:external_factors} implies that uncontrollable external factors cannot ultimately \emph{overcome} the initial bias $u_{\mathrm{o},v}$, but they can only cause slight changes in the individual opinion. Accordingly, the update of agents' opinions over time is modeled through the following difference equation:
\begin{subequations}\label{eq:dyn}
\begin{equation}
    \boldsymbol{x}(t+1)=[\boldsymbol{\lambda}]\boldsymbol{P}\boldsymbol{x}(t)+(\boldsymbol{I}-[\boldsymbol{\lambda}])\boldsymbol{u}(t),~~t \in \mathbb{Z}_{\geq 0},
\end{equation}
where $\boldsymbol{x}(t) \in \mathbb{R}^{|\mathcal{V}|}$ groups the inclinations of all agents in the network, 
\begin{equation}\label{eq:overall_input}
    \boldsymbol{u}(t)=\boldsymbol{u}_{\mathrm{o}}+\boldsymbol{u^{\mathrm{c}}}(t)+\boldsymbol{u^{\mathrm{nc}}}(t),~~t \in \mathbb{Z}_{\geq 0},
\end{equation}
\end{subequations}
and $[\boldsymbol{\lambda}] \in \mathbb{R}^{|\cal{V}| \times |\cal{V}|}$ is a diagonal matrix, whose diagonal elements $\lambda_{vv}=\lambda_{v} \in [0,1]$ are the relative weight of social influence. 

From now on, to further characterize our model, we make two additional assumptions. The following hypothesis is made on the topology of the influence network associated with $\boldsymbol P$.
\begin{assumption}[Network topology]\label{ass:P}
For any node $v \in\mathcal{V}$, there exists a path from it to a node $s\in\mathcal{V}$ such that $\lambda_s<1$. 
\end{assumption}
\begin{remark}[Heterogeneous updates] 
The model \eqref{eq:dyn} assumes homogeneous state updates, i.e., randomness is introduced exclusively from the fact that agents react in a random manner based on the latent opinion. In reality, nonetheless, individuals can also interact in a random manner. This random behavior, whose impact we aim to explore in future works, can be included in the weight matrix $\boldsymbol{P}$.    
\end{remark}
Meanwhile, a second assumption is made on what can be measured, e.g., by a policymaker. Indeed, it is unlikely that one can directly access the individuals' opinions $\boldsymbol{x}(t)$. Instead, it is more realistic to have access to a binary \emph{acceptance variable} $\boldsymbol{y}(t) \in \{0,1\}^{|\mathcal{V}|}$, indicating whether the target choice has been accepted or not, assumed to be linked to $\boldsymbol{x}(t)$ as follows. 

\begin{assumption}[Acceptance variables]\label{ass:Y}
Given the initial condition $x_v(0)\in[0, 1]$, the acceptance variables $y_v(t)\in\{0, 1\}$ are random variables with conditional distribution
\begin{align*}&\mathbb{P}(y_v(t)\!=\!1|\boldsymbol{x}(t))\!=\!x_v(t),~~\mathbb{P}(y_v(t)\!=\!0|\boldsymbol{x}(t))\!=\!1\!-\!x_v(t),
\end{align*}
for all $t \in \mathbb{Z}_{\geq 0}$ and $v \in \cal{V}$.
\end{assumption}

Due to the uncertainty induced by uncontrollable external factors, both the latent opinion $\boldsymbol{x}(t)$ and the acceptance variable $\boldsymbol{y}(t)$ continue to oscillate over time.

\subsection{Asymptotic behavior in the absence of external control}
We now briefly characterize the evolution of the proposed model in \eqref{eq:dyn} assuming that no control action is applied to nudge individuals toward the target choice, i.e., for $\boldsymbol{u}^{\mathrm{c}}(t)=0$ for all $t\in\mathbb{Z}_{\geq 0}$. In this scenario, as a consequence of~\eqref{eq:dyn}, the evolution of the expected dynamics converges to a final limit profile characterized by the following result.
\begin{proposition}[Expected opinion dynamics]\label{prop:convergence}
Let Assumptions~\ref{ass:P}-\ref{ass:Y} be satisfied. Then, for every initial condition $ \boldsymbol{x}(0)\in [0,1]^{\mathcal{V}}$, the dynamics in \eqref{eq:dyn} is such that 
\begin{subequations}
    \begin{equation}
        \mathbb{E}[\boldsymbol{x}(t+1)]=[\boldsymbol{\lambda}]\boldsymbol P\mathbb{E}[\boldsymbol{x}(t)]+(\boldsymbol{I}-[\boldsymbol{\lambda}])\boldsymbol{u}_{\mathrm{o}},~~t\in\mathbb{Z}_{\geq 0},
    \end{equation}
    and 
    \begin{equation}
\label{eq:equilibria}
\boldsymbol{x^{\star}}=\mathbb{E}[{\boldsymbol x}(\infty)]=(I-[\boldsymbol\lambda]\boldsymbol P)^{-1}(I-[\boldsymbol\lambda])\boldsymbol{u}_{\mathrm{o}}.
\end{equation} \hfill $\square$
\end{subequations}		
\end{proposition}
This result can be proven straightforwardly using standard techniques (see \cite{frasca2013gossips}) and, thus, we omit the details. Moreover, it is worth pointing out that Proposition~\ref{prop:convergence} implies that the social system is stable, resulting in the emergence of a final opinion limit profile shaped by interpersonal influences and their strength (through $\boldsymbol{P}$ and $[\boldsymbol{\lambda}]$) and by individual biases. 

Relying on Definition~\ref{def:ergodicity}, we further characterize the ergodicity of the latent and acceptance processes $\boldsymbol{x}(t)$ and $\boldsymbol{y}(t)$, with $t \in \mathbb{Z}_{\geq 0}$, as follows.
\begin{conjecture}[Ergodicity of latent and manifest opinions]
\label{thm:ergodic}
Consider the random processes $\{\boldsymbol x(t)\}_{t\in\integernonnegative}$ and $\{\boldsymbol y(t)\}_{t\in\integernonnegative}$, respectively defined through \eqref{eq:dyn} and Assumption~\ref{ass:Y}. Then, the following holds: 
\begin{enumerate}
\item for $t \rightarrow \infty$, $\boldsymbol x(t),\boldsymbol y(t)$ converge in distribution to two random variables $\boldsymbol x(\infty)$ and $\boldsymbol y(\infty)$;
\item the processes are {\em ergodic};
\item the limit random variables $\boldsymbol x(\infty)$ and $\boldsymbol y(\infty)$ satisfy
$$
\Exp[\boldsymbol x(\infty)]=\Exp[\boldsymbol y(\infty)]=\boldsymbol x^{\star}.
$$ 
\end{enumerate}
\end{conjecture}
This result is not obvious, and the results derived in \cite{ravazzi2015ergodic} using iterated random functions cannot be directly applied. This is because the random maps, describing the evolution, are not identically distributed at each time; instead, their distributions depend on the state.

\subsection{An example of free opinion evolution}\label{sec:example_ol}
\begin{figure}[!tb]
  \centering
\includegraphics[scale=.4,trim=0cm 3cm 0cm 2cm,clip]{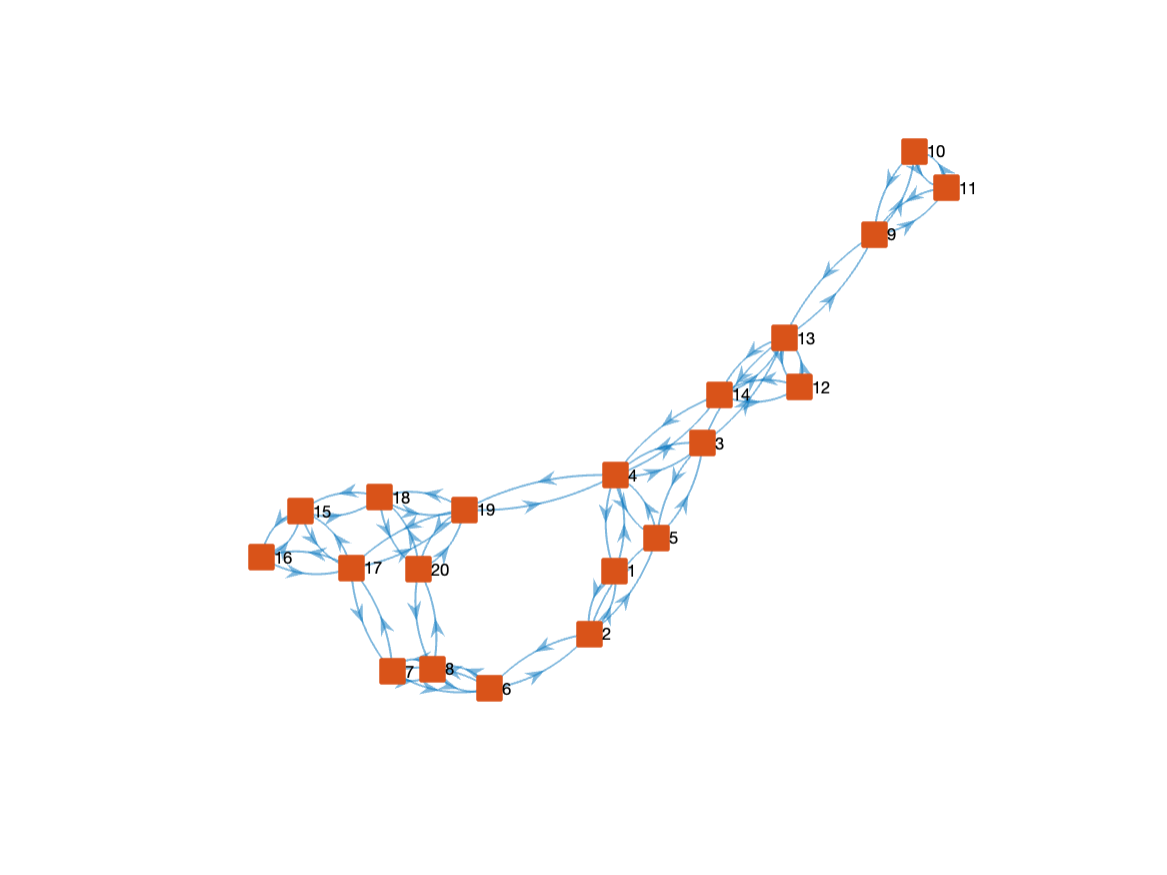} 
  \caption{The clustered social network with $20$ agents considered in our examples.}
  \label{fig:example}
\end{figure}
\begin{figure}[!tb]
\begin{center}
\includegraphics[width=0.99\columnwidth]{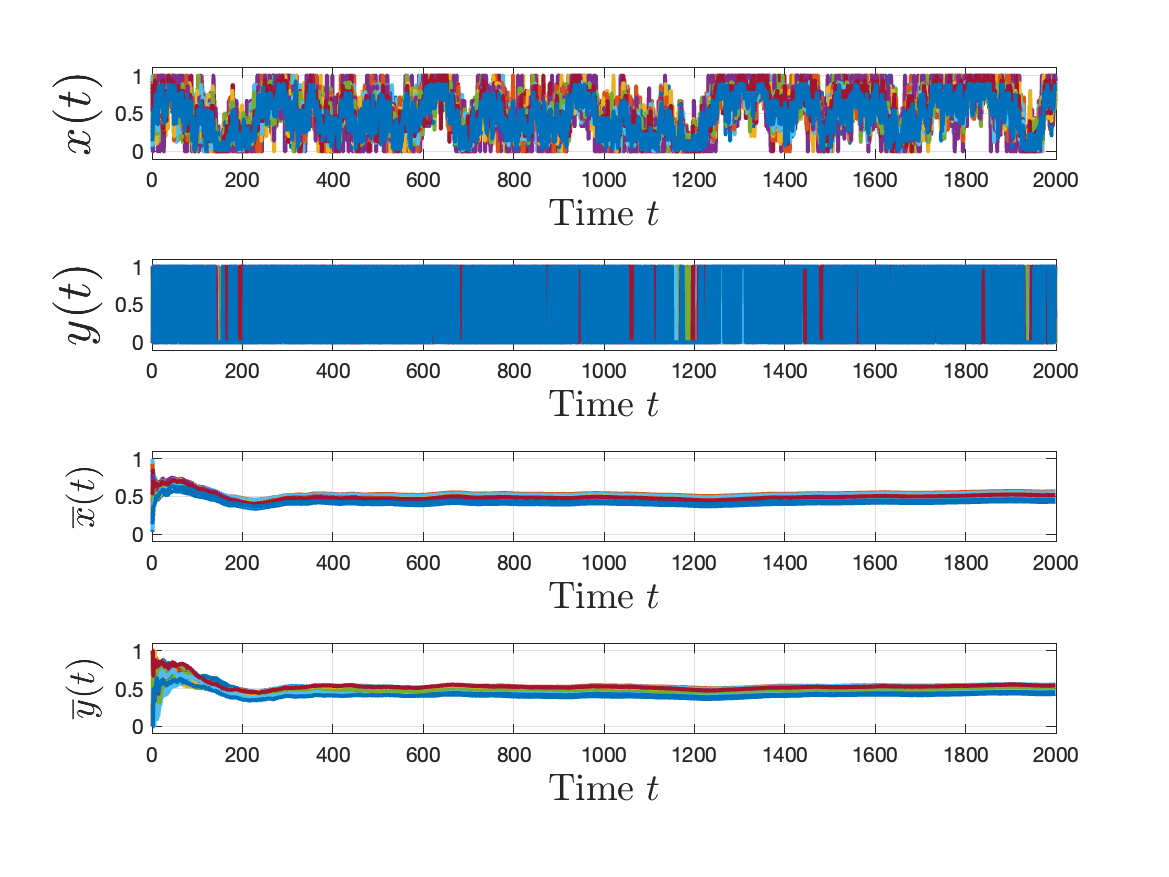}\vspace{-.6cm}
\caption{Opinion dynamics in the considered network of $20$ individuals.}\label{fig:ex1}
\end{center}
\end{figure}
To illustrate what the previous formal properties imply in practice, we consider the network consisting of $n=20$ nodes featuring $7$ clusters, shown in \figurename{~\ref{fig:example}}. The topology of this network is the realization of an Erd\"os Renyi graph with a connection probability $0.2$ and a probability of connections among nodes in different clusters equal to $0.7$. The initial inclination of the $20$ agents is drawn randomly from a uniform distribution in $[0,1]$, while the relative weights of social influence are chosen here to be equal for all agents, i.e., $\lambda_{v}=\lambda$ for all $v=1,\ldots,20$. Meanwhile, the initial acceptance variables $\boldsymbol{y}(0)$ are realizations of an i.i.d. Bernoulli with parameter $1/2$. \figurename{~\ref{fig:ex1}} show the evolution of the latent and manifest opinions of the $20$ agents, as well as their averages, over time. As expected based on our theoretical analysis, the two former random processes exhibit an oscillating behavior, while their averages tend to stabilize to a certain constant value after an initial transient.   

\section{Optimal design of policy interventions}\label{sec:control}
Based on the model introduced in Section~\ref{sec:model}, we now consider the problem of designing cost-efficient and effective policies to nudge individuals toward a target decision. To this end, given the stochastic nature of the latent and manifest opinions' evolutions, we focus on steering individuals toward the acceptance (at least on average) of a target choice by formulating the following optimal control problem
\begin{subequations}\label{eq:problem1}
    \begin{equation}
        \begin{aligned}
            & \min_{\boldsymbol{U^\mathrm{c}}}~~J_{\infty}(\boldsymbol{U^\mathrm{c}})\\
            &~~~ \mbox{s.t.}~~\eqref{eq:dyn},\\
            &\qquad~~~\boldsymbol{u^\mathrm{c}}(t) \in \mathcal{U},~~\forall t \in \mathbb{Z}_{\geq 0},\\ 
            &\qquad~~~\boldsymbol{x}(0)=\boldsymbol{x}_{\mathrm{o}}, 
        \end{aligned}
    \end{equation}
    where $\boldsymbol{x}_{\mathrm{o}} \in [0,1]^{|\cal{V}|}$ comprises the initial opinion of the agents, $\boldsymbol{U^\mathrm{c}}=\{\boldsymbol{u^\mathrm{c}}(t)\}_{t \in \mathbb{Z}_{\geq 0}}$ collects all controllable inputs, the constraint on the controlled input $\mathcal{U}$ is selected to guarantee that $\boldsymbol{u}(t)$ defined in \eqref{eq:overall_input} lays in $[0,1]$ (see Remark~\ref{rem:input_bounds1}), and the loss function is defined as
    \begin{equation}\label{eq:infinite_cost1}
    J_{\infty}(\boldsymbol{U^\mathrm{c}})=\sum_{t=0}^{+\infty} \mathbb{E}[\|\boldsymbol{1}-\boldsymbol{y}(t)\|_{2}^{2}]+r\|\boldsymbol{u}^\mathrm{c}(t)\|_{2}^{2},
\end{equation}
with $r >0$ being a penalty to be tuned by the policymaker to control the relative importance of a widespread acceptance of the target choice across the network and of keeping the policy effort at the minimum level.
\end{subequations}
Nonetheless, the cost function in \eqref{eq:infinite_cost1} can be alternatively rewritten as an \textquotedblleft economic\textquotedblright \ cost depending only on the controlled inputs $\boldsymbol{U^\mathrm{c}}$ and the average latent variables $\{\bar{\boldsymbol{x}}(t)\}_{t \in \mathbb{Z}_{\geq 0}}$, as detailed in the following proposition.   
\begin{proposition}\label{prop:change_cost}
The cost function in \eqref{eq:infinite_cost1} can be equivalently rewritten as  
\begin{equation}\label{eq:infinite_cost2}
    J_{\infty}(\boldsymbol{U^\mathrm{c}})=\sum_{t=0}^{+\infty} \|\boldsymbol{1}-\bar{\boldsymbol{x}}(t)]
    \|_{1}+r\|\boldsymbol{u}^\mathrm{c}(t)\|_{2}^{2},
\end{equation}
where $\bar{\boldsymbol{x}}(t)=\mathbb{E}[\boldsymbol{x}(t)]$.
\end{proposition}
\begin{proof}
The proof can be found in Appendix~\ref{appendix:proof1}
\end{proof}
This result implies that the expected acceptance of the target choice can be maximized while minimizing the control effort by only looking at the expected latent variables $\bar{\boldsymbol{x}}(t)$ over time. This simplification is possible at the price of replacing a quadratic cost in $\boldsymbol{y}(t)$ with a 1-norm penalty in $\boldsymbol{1}-\bar{\boldsymbol{x}}(t)$ for all $t \in \mathbb{Z}_{\geq 0}$, ultimately shrinking the latent variables to full acceptance (i.e., $\boldsymbol{1}$). 

Despite the result proven in Proposition~
\ref{prop:change_cost}, we postpone the design of a policy with the economic loss in \eqref{eq:infinite_cost2} to future work. Instead, in this paper, we propose two approximations that allow us to get back to Quadratic Programs (QPs) for policy design. In introducing these approximations, we first assume that we can measure the expected latent variable, and then relax this (unrealistic) assumption.  

\begin{remark}[Input feasibility set]\label{rem:input_bounds1}
    In the absence of external (uncontrolled) inputs, the controlled input is guaranteed to satisfy $\boldsymbol{u}(t) \in [\boldsymbol{0},\boldsymbol{1}]$ for all $t \in \mathbb{Z}_{\geq 0}$ (see \eqref{eq:overall_input}) if $\mathcal{U}$ in \eqref{eq:problem1} corresponds to $[\boldsymbol{0},\boldsymbol{1}-\boldsymbol{u_{\mathrm{o}}}]$. Unfortunately, such a choice might not yield a feasible input when external disturbances influence the decision process. While we postpone the formulation of probabilistic input constraints to future works, from now, we consider the following (worst case) set of feasible inputs:
    \begin{equation}\label{eq:input_constraint}
        \mathcal{U}\!=\![\boldsymbol{0},\boldsymbol{1}\!-\boldsymbol{u_{\mathrm{o}}'}] \subset [\boldsymbol{0},\boldsymbol{1}\!-\boldsymbol{u_{\mathrm{o}}}],~~\mbox{with}~ \boldsymbol{u_{\mathrm{o}}'}=\boldsymbol{u_{\mathrm{o}}}+\sqrt{\frac{(\delta\boldsymbol{u}_{0})^{2}}{3}},
    \end{equation}
    namely, we shrink the set of feasible inputs considering the standard deviation of the external disturbances (according to Assumption~\ref{ass:external_factors}). 
\end{remark}
\begin{remark}[Nudging decisions \& ethical issues]
  Policies designed by solving \eqref{eq:problem1} or its reformulations proposed in the following subsection could be leveraged by a policymaker to promote virtuous behaviors, e.g., the use of shared mobility services to reduce pollution. At the same time, they are tightly linked to the policymaker's design choices and intent. Therefore, optimal control approaches like \eqref{eq:problem1} can also be misused. We thus wish to acknowledge the possible ethical issues consequent to the use of control techniques to nudge individual choices, which will be the subject of future investigations.
\end{remark}

\begin{remark}[Optimality over an infinite horizon]\label{remark:horizons}
Losses like the (equivalent) ones in \eqref{eq:infinite_cost1} and \eqref{eq:infinite_cost2} often do not represent something that a policymaker can and actually wants to minimize in practice. Indeed, solving policy design problems based on these losses leads to actions that are \textquotedblleft optimal\textquotedblright \ over an infinite horizon but disregard possible changes induced by not modeled, yet important, aspects, e.g., technological advances that can significantly change one's inclination if the target choice is whether or not to adopt a new technology. Actual policies will instead adapt to these factors, making it more reasonable to consider finite horizon costs and solve the design problem in a \emph{receding horizon fashion} (as we will do in the following Subsections).
\end{remark}

\subsection{Conservative policy design}
As a first strategy to simplify the economic loss in \eqref{eq:infinite_cost2}, we propose to replace it with its lower bound\footnote{The lower bound stems from the properties of norms and the fact that $\boldsymbol{1}-\bar{\boldsymbol{x}}(t) \in [0,1]^{|\cal{V}|}$.}
\begin{equation}\label{eq:cost_approx}
    \hat{J}_{\infty}(\boldsymbol{U^\mathrm{c}})=\sum_{t=0}^{+\infty} \|\mathbf{1}-\bar{\boldsymbol{x}}(t)\|_{2}^{2}+r\|\boldsymbol{u}^\mathrm{c}(t)\|_{2}^{2},
\end{equation}
hence shifting to a \emph{conservative} quadratic loss. Despite this simplification, as pointed out in Remark~\ref{remark:horizons}, looking at an infinite horizon loss might not be ideal when designing policies to nudge individual opinions. Therefore, we further shift from an infinite-horizon cost to a finite-horizon one. This choice allows us to formulate the following design problem on the mean opinion dynamics:
\begin{equation}\label{eq:conservative_MPC}
    \begin{aligned}
        & \min_{\boldsymbol{U^\mathrm{c}}_{T}} \hat{J}_{T}(\boldsymbol{U^\mathrm{c}_{T}})\\
        & ~~\mbox{s.t. }~\bar{\boldsymbol{x}}(k\!+\!1|t)\!=\![\boldsymbol{\lambda}]\boldsymbol{P}{\bar{\boldsymbol{x}}}(k|t)\!+\!(\boldsymbol{I}\!-\![\boldsymbol{\lambda}])(\boldsymbol{u}_{\mathrm{o}}\!+\!\boldsymbol{u}^\mathrm{c}(k|t)),\\
        & ~~\qquad {\boldsymbol{u}^\mathrm{c}}(k|t) \in \mathcal{U},~~~k \in [0,T-1],\\
        & ~~\qquad \bar{\boldsymbol{x}}(0|t)=\bar{\boldsymbol{x}}(t),\\
        & ~~\qquad \|\mathbf{1}\!-\!\bar{\boldsymbol{x}}(k\!+\!1|t)\|_{1} \!\leq \!\alpha \|\mathbf{1}\!-\!\bar{\boldsymbol{x}}(k|t)\|_{1},~k \!\in\! [0,T\!-\!1],
    \end{aligned}
\end{equation}
where $T \geq 1$  is the policy design horizon chosen by the policymaker, $\boldsymbol{U^\mathrm{c}}_{T}=\{\boldsymbol{u}^\mathrm{c}(k|t)\}_{k=0}^{T-1}$ is the set of optimal actions designed based on the mean inclinations comprised in $\bar{\boldsymbol{x}}(t)$, $\bar{\boldsymbol{x}}(k|t)$ is the predicted state according to \eqref{eq:dyn} and the initial condition dictated by $\bar{\boldsymbol{x}}(t)$, for $k \in [0,T-1]$. Meanwhile, the cost function is given by
\begin{equation}\label{eq:cost_1}
    \hat{J}_{T}(\boldsymbol{U^\mathrm{c}_{T}})=\sum_{k=0}^{T-1} \|\mathbf{1}-\bar{\boldsymbol{x}}(k|t)\|_{2}^{2}+r\|{\boldsymbol{u}}^\mathrm{c}(k|t)\|_{2}^{2},
\end{equation}
the constraint on the controlled input is dictated by $\mathcal{U}$ as defined in \eqref{eq:input_constraint}, and, for $\alpha<1$, the last (contraction) constraint in \eqref{eq:conservative_MPC} enforces the predictive tracking error to reduce at each time instant (similar to the one proposed in \cite{POLAK1993}). As also remarked in \ref{remark:horizons}, this problem should be solved in a receding-horizon fashion at each policy design step $t \in \mathbb{Z}_{\geq 0}$, keeping track and, hence, adapting to changes in the mean inclination of agents with respect to the target choice.

\subsection{An alternative loss to track opinion shifts}\label{sec:alternative_map}
To find an alternative to the previous approximation of the initial policy design problem (see \eqref{eq:problem1}), we now exploit the fact that if $\bar{\boldsymbol{x}}(t)$ \emph{converges}, i.e.,
$$
|\bar{x}_{v}(t+1) -\bar{x}_{v}(t)|\underset{t \rightarrow \infty}{\longrightarrow} 0,~~\forall v \in \mathcal{V},
$$
then the first term of the economic loss in \eqref{eq:infinite_cost1} can be approximated as
\begin{equation}
    \|1-\bar{\boldsymbol{x}}(t)\|_{1} \underset{t \rightarrow \infty}{\longrightarrow} \sum_{v \in V}\frac{(1-\bar{x}_{v}(t))^2}{|1-\bar{x}_{v}(t-1)|}.
\end{equation}
Relying on this result, we then recast the finite-horizon problem in \eqref{eq:conservative_MPC} as follows 
\begin{equation}\label{eq:TV_MPC}
    \begin{aligned}
        & \min_{\boldsymbol{U^\mathrm{c}}_{T}} \tilde{J}_{T}(\boldsymbol{U^\mathrm{c}_{T}})\\
        & ~~\mbox{s.t. }~\bar{\boldsymbol{x}}(k\!+\!1|t)\!=\![\boldsymbol{\lambda}]\boldsymbol{P}{\bar{\boldsymbol{x}}}(k|t)\!+\!(\boldsymbol{I}\!-\![\boldsymbol{\lambda}])(\boldsymbol{u}_{\mathrm{o}}\!+\!\boldsymbol{u}^\mathrm{c}(k|t)),\\
        & ~~\qquad {\boldsymbol{u}^\mathrm{c}}(k|t) \in \mathcal{U},~~~k \in [0,T-1],\\
        & ~~\qquad \bar{\boldsymbol{x}}(0|t)=\bar{\boldsymbol{x}}(t),\\
        & ~~\qquad \|\mathbf{1}\!-\!\bar{\boldsymbol{x}}(k\!+\!1|t)\|_{1} \!\leq \!\alpha \|\mathbf{1}\!-\!\bar{\boldsymbol{x}}(k|t)\|_{1},~k \!\in\! [0,T\!-\!1],
    \end{aligned}
\end{equation}
where
\begin{equation}\label{eq:cost_approx2}
    \tilde{J}_{T}(\boldsymbol{U^\mathrm{c}})=\sum_{k=0}^{T-1} \|\mathbf{1}-\bar{\boldsymbol{x}}(k|t)\|_{\boldsymbol{Q}(t+k)}^{2}+r\|{\boldsymbol{u}}^\mathrm{c}(k|t)\|_{2}^{2},
\end{equation}
where $\boldsymbol{Q}(t+k)$ is a diagonal matrix, with\footnote{To avoid numerical problems, the weight $Q_{v}(t+k)$ can be practically set to $Q_{v}(t+k)=\left(|1-\bar{x}_{v}(t+k-1)|+\varepsilon \right)^{-1}$, with $\varepsilon>0$ being a small constant value guarantee that the weights are well-posed.}
\begin{equation}\label{eq:original_weights}
Q_{v}(t+k)=\left(|1-\bar{x}_{v}(t\!+\!k\!-\!1)| \right)^{-1},~~k \in [0,T-1], 
\end{equation}
for all $t \in \mathbb{Z}_{\geq 0}$ and for all $v \in \mathcal{V}$. 

Clearly, this reformulation requires adjusting the penalties in the control cost over time based on the \textquotedblleft future\textquotedblright \ expected latent variables. On the one hand, this design choice would allow us to progressively recover the finite-horizon counterpart of the original economic cost if both the controlled input and the mean dynamics converge to a steady-state\footnote{Future work will be devoted to analyzing the conditions under which this result holds and, hence, understanding when solving \eqref{eq:TV_MPC} can be worthwhile.} (see Appendix~\ref{appendix:additional_details} for additional details). On the other hand, the dependence on the \textquotedblleft future\textquotedblright \ (and, thus, unknown at design time) expected latent variables makes the design problem more complex. Indeed, solving \eqref{eq:TV_MPC} would require a preview on $\bar{\boldsymbol{x}}(t+k)$, for $k=0,\ldots,T-1$. To circumvent this limitation, we propose to \textbf{do} the following at each time instant $t \in \mathbb{Z}_{\geq 0}$:
\begin{enumerate}
    \item Exploit the optimal sequence computed at the previous instant $t$ to construct the \textquotedblleft candidate\textquotedblright \ input sequence\footnote{At the first time step $\boldsymbol{U_{T}^\mathrm{c}}$ can be set to zero.} $$\boldsymbol{U_{T}^\mathrm{c}}=\{\boldsymbol{u}^{c,\star}(t|t-1),\ldots,\boldsymbol{u}^{c,\star}(t+T-1|t-1),0\},$$ which satisfies by construction the input constraints featured in .
    \item Use $\boldsymbol{U_{T}^\mathrm{c}}$ to predict the evolution of the mean inclination $\{\boldsymbol{\bar{x}}^{p}(t+k)\}_{k=1}^{T-1}$ through our mean inclination model, i.e.,
\begin{align*}
  \boldsymbol{\bar{x}}^{p}(t+k+1)&=[\boldsymbol{\lambda}]\boldsymbol{P}\boldsymbol{\bar{x}}^{p}(t+k)\\
  &~~~+(I-[\boldsymbol{\lambda}])(\boldsymbol{u}_{\mathrm{o}}+\boldsymbol{u}^{c,\star}(t+k|t-1)),
\end{align*}
for all $k \in [0,T-1]$, with $\boldsymbol{\bar{x}}^{p}(t)=\boldsymbol{\bar{x}}(t)$.
\item Approximate the weights in \eqref{eq:original_weights} as
\begin{equation}
    \hat{Q}_{v}(t+k)\!=\!(|1\!-\!\bar{x}_{v}^{p}(t+k-1)|)^{-1},~~\forall k \in [0,T-1],
\end{equation}
setting $\boldsymbol{\bar{x}}^{p}(t-1)+\boldsymbol{\bar{x}}(t-1)$.
\end{enumerate}  
\subsection{Policy design with estimated mean inclinations}
All the schemes proposed up to this moment rely on the exact knowledge of $\bar{\boldsymbol{x}}(t)$ for all $t \in \mathbb{Z}_{\geq 0}$. However, this quantity is generally not accessible to policymakers that, instead, can usually access the manifest opinion modeled here through $\boldsymbol{y}(t)$ (see Assumption~\ref{ass:Y}). These variables can nonetheless be used to define a practical estimate of the mean inclination over time as
\begin{equation}\label{eq:estimate}
    \hat{\bar{\boldsymbol{x}}}(t)=\frac{1}{t}\sum_{\tau=0}^{t-1} \vy(\tau).
\end{equation}
In fact, if the ergodicity of both the latent and manifest opinions proven in Proposition~\ref{prop:convergence} is preserved under the designed control action, then we expect that the Ces\'aro averages converge to their expected value, i.e.,
\begin{equation*}
    \mathbb{E}[\|\hat{\bar{\boldsymbol{x}}}(t)-\bar{\boldsymbol{x}}(t)\|_2^2]\underset{t \rightarrow \infty}{\longrightarrow}0.
\end{equation*}
This estimate is used to replace the actual value of the mean latent variable in the initial condition of both \eqref{eq:conservative_MPC} and \eqref{eq:TV_MPC}, leading to two problems that can be practically solved by a policymaker to design interventions to nudge individuals toward the target choice.  
\section{Numerical example}\label{sec:numerical}
We consider the same network introduced in Section~\ref{sec:example_ol} (see \figurename{~\ref{fig:example}}), still assuming $\lambda_v=\lambda$ for all $v \in \cal{V}$. Despite this simplification, we analyze two different cases,  $\lambda=0.25$ and $\lambda=0.75$, to consider two opposite situations. The former entails that the agents change their inclination mainly based on their bias and external interventions, while the latter implies that the main drivers of one's inclination are social influences. Apart from the initial conditions (they are drawn randomly from a uniform distribution in $[\boldsymbol{0},\boldsymbol{1}]$), the other source of differentiation between agents is represented by the values of $\boldsymbol{u}_{\mathrm{o}}$, for which we consider four different scenarios.
\begin{enumerate}
    \item Firstly, we consider $u_{\mathrm{o},v}=0.7$ for all $v \in \V$, which implies a pre-existing positive bias within the population.
    \item Then, we consider the case in which $10$ agents have $u_{\mathrm{o},v}=0.2$ and the remaining 10 have $u_{\mathrm{o},v}=0.8$. This implies that the population's bias is skewed, with half of it represented by what can be seen as less receptive agents. 
    \item As a third scenario, we consider a heterogeneous population, where individuals are generally positively biased. In this case, half of the agents (i.e., $10$) are characterized by $u_{\mathrm{o},v}=0.6$ and the remaining half has an initial bias of $0.8$.
    \item At last, we consider the completely opposite situation of an overall more receptive population with half of the agents having an initial bias of 0.2 and the remaining ones having $u_{\mathrm{o},v}=0.3$.
\end{enumerate}

We design and apply the proposed policies over a horizon of $T=30$ steps. To evaluate its outcome, we consider three quality indicators. As an empirical check of recursive feasibility, we evaluate (if any) the number of time instant $\tau^{\mathrm{ob}}$ in which the actual (inaccessible) inclinations exceed the interval $[\boldsymbol{0},\boldsymbol{1}]$. This indicator allows us to assess the benefit of the \emph{worst-case} approach adopted to cope with (unpredictable) external factors. As a second quality indicator, we consider the average number of adopters $\Gamma_{T}$ over time, i.e.,
\begin{equation}
\Gamma_{T}=\frac{1}{|\V|T}\sum_{t=0}^{T-1}\sum_{v \in V}y_{v}(T)\cdot 100, ~~[\%]
\end{equation}
which ultimately provides us with an indication of the \emph{effectiveness} of each policy strategy at the end of the considered horizon. An actual cost-benefit analysis cannot be carried out without considering an additional indicator, i.e.,
\begin{equation}
    \Delta u^\mathrm{c}_{T}=\sum_{t=0}^{T-1} \|\boldsymbol{u^\mathrm{c}}(t)\|_1,
\end{equation}
which provides insight into the overall effort required to enact a given policy. In the following, we denote the conservative policy obtained by solving the worst-case problem (see \eqref{eq:conservative_MPC}) as (WC), and the time-varying solution as (TV). In both these cases, the state is fully accessible. When estimating the state, we solely test the latter policy, indicating its outcomes and the associated performance indexes through (E-TV). To design all these policies\footnote{Future work will be devoted to analyzing the impact of the tuning parameters on policy behavior.}, we set $r=0.1$, $\alpha=0.99$ and $\delta=0.025$.  

\subsection{Policy impact in the noise-free case}
\begin{table*}[!tb]
\caption{Disturbance-free case: indicators for $\lambda=0.25$}\label{tab:025}
\centering
\begin{tabular}{ccccccc}
\multicolumn{1}{c}{} & $\Gamma_{T}$ (WC) & $\Gamma_{T}$ (TV) &  $\Gamma_{T}$ (E-TV) & $\Delta u^\mathrm{c}_{T}$ (WC) & $\Delta u^\mathrm{c}_{T}$ (TV) & $\Delta u^\mathrm{c}_{T}$ (E-TV)\\
\hline
Scenario 1 & \textbf{94.83}  \%& 87.24 \% & 88.10 \% & 5.29 & 3.70 &\textbf{3.54}\\
\hline
Scenario 2 & \textbf{92.76} \% & 81.72 \% & 81.38 \% & 8.78 & 6.67 & \textbf{6.43}\\
\hline
Scenario 3 & \textbf{93.62} \% & 85.34 \% & 88.10 \% & 5.29 & 3.78 & \textbf{3.72}\\
\hline
Scenario 4 & \textbf{88.96} \% & 76.90 \% & 79.48 \% & 13.23 & \textbf{10.06} & \textbf{10.06}\\
\hline
\end{tabular}
\end{table*}
\begin{table*}[!tb]
\caption{Disturbance-free case: indicators for $\lambda=0.75$}\label{tab:075}
\centering
\begin{tabular}{ccccccc}
\multicolumn{1}{c}{} & $\Gamma_{T}$  (WC) & $\Gamma_{T}$  (TV) &  $\Gamma_{T}$  (5-TV) & $\Delta u^\mathrm{c}_{T}$ (WC) & $\Delta u^\mathrm{c}_{T}$ (TV) & $\Delta u^\mathrm{c}_{T}$ (E-TV)\\
\hline
Scenario 1 & \textbf{91.38} \% &  88.10 \%& 90.17 \%& 5.30 & \textbf{4.66} &5.44\\
\hline
Scenario 2 & \textbf{89.13} \% & 86.03 \% & \textbf{89.13} \% & 8.64 & \textbf{7.99} & 8.75\\
\hline
Scenario 3 & \textbf{92.59} \% & 87.07 \% & 91.90 \% & 5.27 & \textbf{4.69} & 5.35\\
\hline
Scenario 4 & 88.79 \%& 87.93 \% & \textbf{88.97} \% & 13.19 & \textbf{ 12.71} & 13.73\\
\hline
\end{tabular}
\end{table*}
\begin{figure}[!tb]
    \begin{tabular}{cc}
        \hspace*{-.5cm} \subfigure[Hidden inclinations (WC)]{\includegraphics[scale=.525]{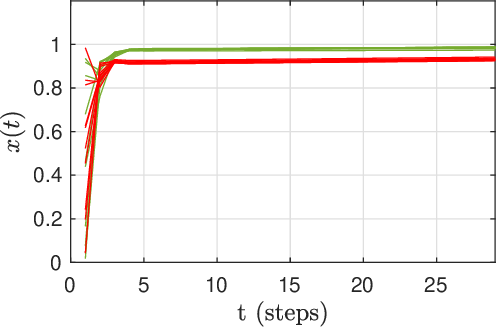}} \hspace*{-.3cm}&\hspace*{-.3cm} \subfigure[Policy actions (WC)]{\includegraphics[scale=.525]{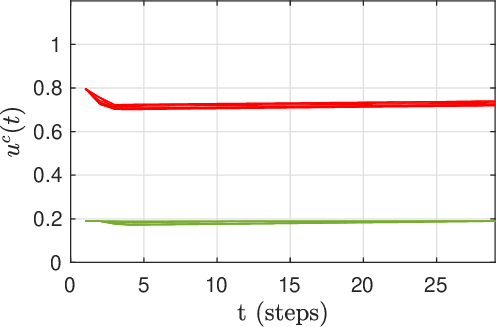}}\\
        \hspace*{-.5cm} \subfigure[Hidden inclinations (E-TV)]{\includegraphics[scale=.525]{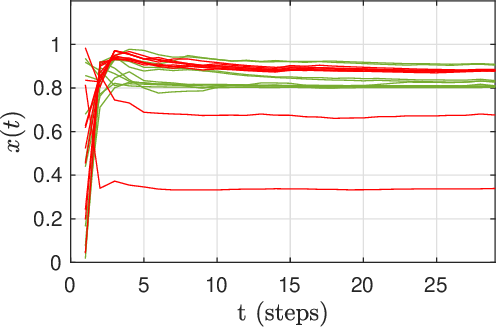}} \hspace*{-.3cm}&\hspace*{-.3cm} \subfigure[Policy actions (E-TV)]{\includegraphics[scale=.525]{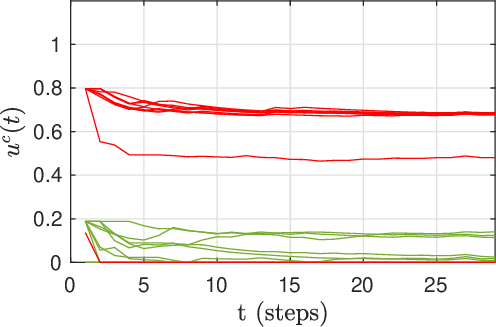}}
    \end{tabular}
    \caption{Disturbance-free case: states and inputs for the best-performing policies, i.e., (WC) and (E-TV) for scenario 2 and $\lambda=0.25$ over time. The green curves are associated with agents with positive bias ($u_{\mathrm{o},v}=0.8$), while the red ones are associated with agents having $u_{\mathrm{o},v}=0.2$.}
    \label{fig:025}
\end{figure}
\begin{figure}[!tb]
    \begin{tabular}{cc}
        \hspace*{-.5cm} \subfigure[Hidden inclinations (WC)]{\includegraphics[scale=.525]{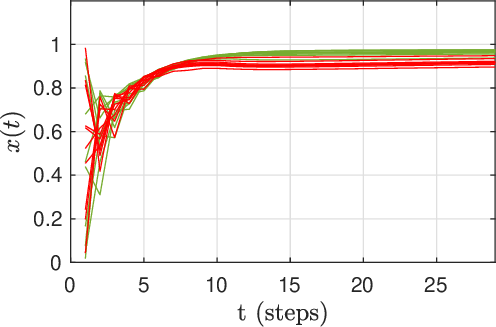}} \hspace*{-.3cm}&\hspace*{-.3cm} \subfigure[Policy actions (WC)]{\includegraphics[scale=.525]{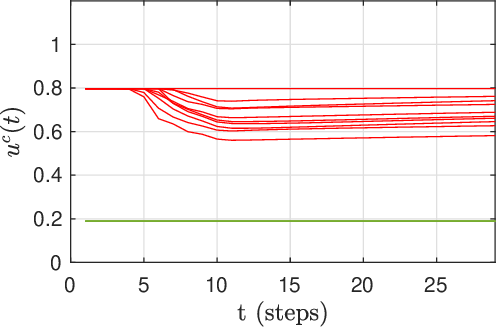}}\\
        \hspace*{-.5cm} \subfigure[Hidden inclinations (TV)]{\includegraphics[scale=.525]{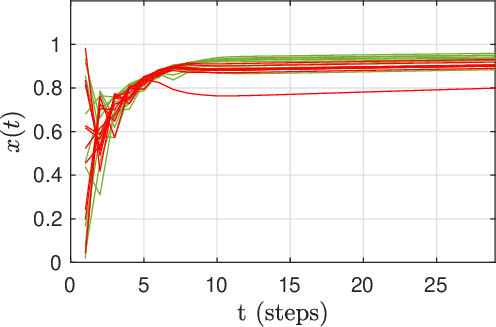}} \hspace*{-.3cm}&\hspace*{-.3cm} \subfigure[Policy actions (TV)]{\includegraphics[scale=.525]{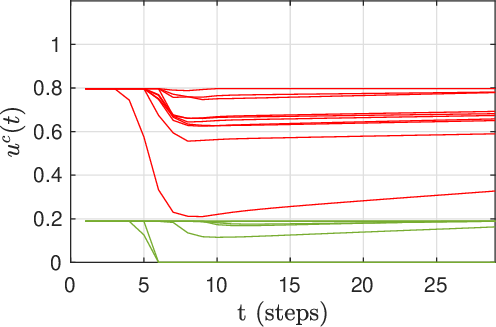}}
    \end{tabular}
    \caption{Disturbance-free case: states and inputs for the best-performing policies, i.e., (WC) and (TV) for scenario 2 and $\lambda=0.75$ over time. The green curves are associated with agents with positive bias ($u_{\mathrm{o},v}=0.8$), while the red ones are associated with agents having $u_{\mathrm{o},v}=0.2$.}
    \label{fig:075}
\end{figure}
Tables~\ref{tab:025}-\ref{tab:075} report the value of the indicators for all the considered scenarios. Independently on the value of $\lambda$, the worst-case conservative law obtained by solving \eqref{eq:conservative_MPC} is the one generally resulting in the most widespread diffusion of a positive inclination. This generally comes at the price of an increased policy cost. From this perspective, the policies with adapting weights tend to outperform the \textquotedblleft brute-force\textquotedblright \ approximation performed in \eqref{eq:cost_approx}. This confirms the capability of the latter policy to better adapt to changes in individual inclinations (despite the approximation performed in constructing the weights) and, hence, allow policymakers to save money. 

Both Tables~\ref{tab:025}-\ref{tab:075} show a general tendency of costs to increase when more agents tend to be receptive to the policy, while inherently negatively biased. As one would have guessed intuitively, containing investments in such a scenario might be counterproductive. Meanwhile, comparing the results in these tables, it can be noticed that all policies tend to be closer in performance and cost for a higher $\lambda$, showcasing the impact that different degrees in social connection can have on the final policy outcome and cost-benefit trade-off. At the same time, these results do not show any specific drop in performance, validating the use of the \textquotedblleft simple\textquotedblright \ estimate \eqref{eq:estimate} to reconstruct the (hidden) state, at least in this relatively simple example.  

Focusing on the \textquotedblleft skewed\textquotedblright \ scenario (i.e., scenario 2), \figurename~{\ref{fig:025}} and \figurename~{\ref{fig:075}} show the evolution of the hidden state (a.k.a., inclinations) of all agents over time, along with the policies enacted to achieve such result. From these plots, the difference between the two conditions on $\lambda$ and, hence, the effect of social interactions is clear. A higher $\lambda$ indeed promotes agents' inclination to evolve more cohesively, with few agents left behind at the price of a slower achievement of a widespread positive inclination. Also when looking at the policy actions, the changes required by different levels of social connections are clear. In particular, when $\lambda=0.75$, policy actions tend to be more varied (especially for negatively biased agents), indicating the need for policymakers to counteract the effect induced by social interaction and individual biases in a more dedicated way. Last but not least, the results obtained when exploiting the estimated state rather than the actual state (see the lower side of \figurename~{\ref{fig:025}}) are still consistent with the general trend observed when the true (but practically unknown state) is available. At the same time, small oscillations appear in the transient of both the state and the input, spotlighting the fact that approximate information is employed for policy design.  

\subsection{Policy impact with external disturbances}
\begin{table*}[!tb]
\centering
\caption{Noisy case: indicators for $\lambda=0.25$}\label{tab:025_2}
\begin{tabular}{ccccc}
\multicolumn{1}{c}{} & $\Gamma_{T}$ (E-WC) &  $\Gamma_{T}$ (E-TV) & $\Delta u^\mathrm{c}_{T}$ (E-WC) & $\Delta u^\mathrm{c}_{T}$ (E-TV)\\
\hline
Scenario 1 & \textbf{96.55} \% & 87.59 \%& 5.37 & \textbf{3.62}\\
\hline
Scenario 2 & \textbf{93.79} \% & 80.17 \%& 8.83 & \textbf{6.51}\\
\hline
Scenario 3 & \textbf{95.51} \% & 88.62 \%&  5.40 & \textbf{ 3.74}\\
\hline
Scenario 4 & \textbf{91.55} \% & 72.41 \%& 13.18 & \textbf{9.63}\\
\hline
\end{tabular}
\end{table*}
\begin{table*}[!tb]
\centering
\caption{Noisy case: indicators for $\lambda=0.75$}\label{tab:075_2}
\begin{tabular}{ccccc}
\multicolumn{1}{c}{} & $\Gamma_{T}$ (E-WC) &  $\Gamma_{T}$ (E-TV) & $\Delta u^\mathrm{c}_{T}$ (E-WC) & $\Delta u^\mathrm{c}_{T}$ (E-TV)\\
\hline
Scenario 1 & \textbf{92.24} \% & 91.38 \%& 5.58 & \textbf{5.38}\\
\hline
Scenario 2 & \textbf{90.69} \% & \textbf{90.68} \%& 9.26 & \textbf{8.85}\\
\hline
Scenario 3 & \textbf{93.10} \% & 90.52 \%&  5.60 & \textbf{5.46}\\
\hline
Scenario 4 & \textbf{90.52} \% & 90.52 \%& 13.66 & \textbf{13.43}\\
\hline
\end{tabular}
\end{table*}
\begin{figure}[!tb]
    \begin{tabular}{cc}
        \hspace*{-.5cm} \subfigure[Hidden inclinations (E-WC)]{\includegraphics[scale=.525]{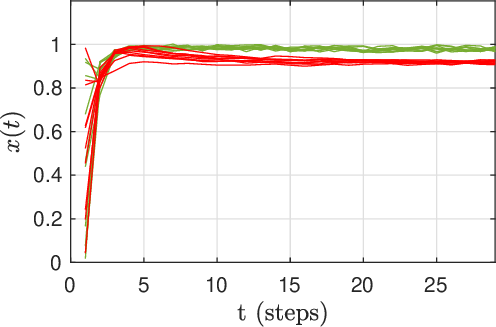}} \hspace*{-.3cm}&\hspace*{-.3cm} \subfigure[Policy actions (E-WC)]{\includegraphics[scale=.525]{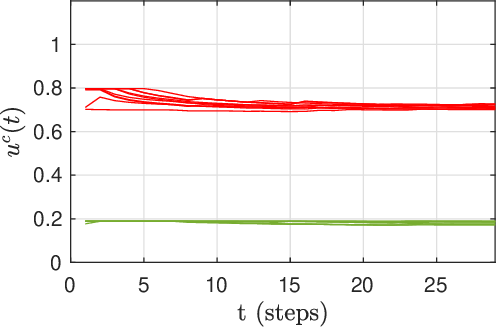}}\\
        \hspace*{-.5cm} \subfigure[Hidden inclinations (E-TV)]{\includegraphics[scale=.525]{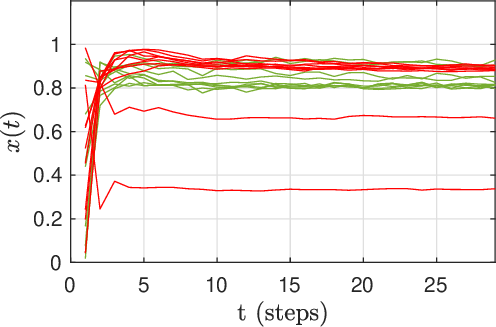}} \hspace*{-.3cm}&\hspace*{-.3cm} \subfigure[Policy actions (E-TV)]{\includegraphics[scale=.525]{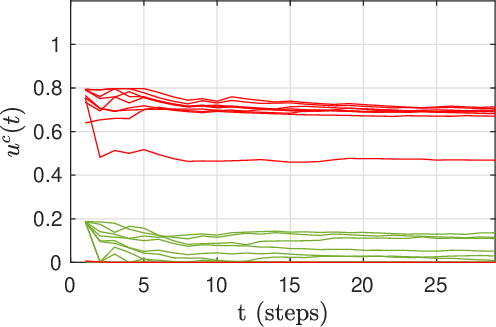}}
    \end{tabular}
    \caption{Noisy case: states and inputs for  (E-WC) and (E-TV) for scenario 2 and $\lambda=0.25$ over time. The green curves are associated with agents with positive bias ($u_{\mathrm{o},v}=0.8$), while the red ones are associated with agents having $u_{\mathrm{o},v}=0.2$.}
    \label{fig:025_2}
\end{figure}
\begin{figure}[!tb]
    \begin{tabular}{cc}
        \hspace*{-.5cm} \subfigure[Hidden inclinations (E-WC)]{\includegraphics[scale=.525]{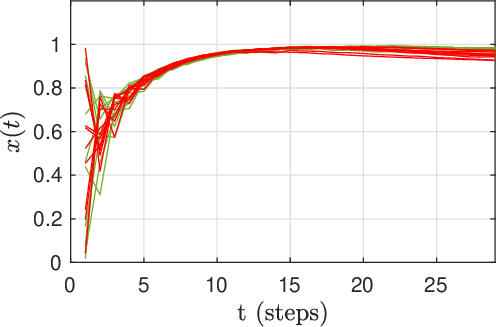}} \hspace*{-.3cm}&\hspace*{-.3cm} \subfigure[Policy actions (E-WC)]{\includegraphics[scale=.525]{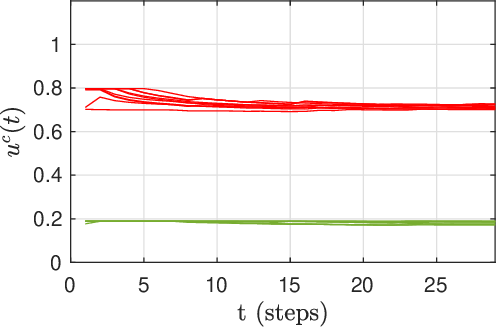}}\\
        \hspace*{-.5cm} \subfigure[Hidden inclinations (E-TV)]{\includegraphics[scale=.525]{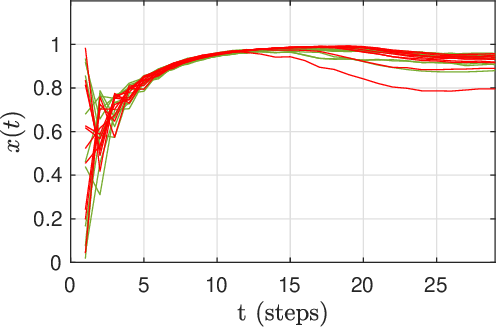}} \hspace*{-.3cm}&\hspace*{-.3cm} \subfigure[Policy actions (E-TV)]{\includegraphics[scale=.525]{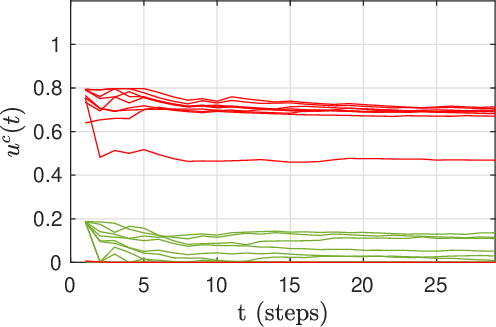}}
    \end{tabular}
    \caption{Noisy case: states and inputs for (E-WC) and (E-TV) for scenario 2 and $\lambda=0.75$ over time. The green curves are associated with agents with positive bias ($u_{\mathrm{o},v}=0.8$), while the red ones are associated with agents having $u_{\mathrm{o},v}=0.2$.}
    \label{fig:075_2}
\end{figure}
When considering the impact of external disturbances, we use the same performance criteria as before. Nonetheless, we focus only on the worst-case policy designed via the mean inclination estimate (E-WC) and the one obtained with such estimate and time-varying weights. The attained performance indicators are reported in Tables~\ref{tab:025_2}-\ref{tab:075_2} for both conditions on $\lambda$, showing that the trends evidenced in the absence of noise are confirmed when (unpredictable) external disturbances impact individual inclinations. At the same time, at least in our example, this shows that the estimate for the mean inclinations can represent an asset in the absence of direct insight into individual beliefs (which are anyway unrealistic in practice). 

Mimicking the analysis performed in the disturbance-free case, in \figurename~{\ref{fig:025_2}} and \figurename~{\ref{fig:075_2}} we show the (hidden) inclination and (controlled) input trajectories obtained for the two considered policies and the two scenarios on $\lambda$. Once again, the results we obtain are consistent with those achieved in the disturbance-free case, despite slightly deteriorated by the uncertainty on the initial condition of the receding horizon controllers and the presence of $\boldsymbol{u^\mathrm{nc}}$.

\section{Concluding Remarks}\label{sec:conclusions}
This work presented a model-based framework to analyze the impact of external interventions on the evolution of individual inclinations subject to social influences and the effects of unpredictable phenomena (i.e., changes in the weather forecast when deciding to use a shared bicycle). This formalization has allowed us to devise a set of optimal policy design strategies looking for a balance between nudging individuals toward a target choice and containing costs. Our numerical analysis spotlighted the differences, benefits, and drawbacks of the proposed strategies, setting the ground for further research into the use of these models for effective policy design in a social setting.

Future work will be devoted to proving the theoretical guarantees on the performance of the proposed policy design schemes, along with considering the full cost without approximations (hence, a stochastic policy design problem), and validating the approaches on real-world applications. 

\appendix
\subsection{Proof of Proposition~\ref{prop:change_cost}}\label{appendix:proof1}
Let us focus on the first term of the sum in equation \eqref{eq:infinite_cost1}, for which the following holds:
$$
\mathbb{E}[\|\boldsymbol{1}\!-\!\boldsymbol{y}(t)\|_{2}^{2}]= \mathbb{E}[\mathbb{E}[\|\boldsymbol{1}-\boldsymbol{y}(t)\|_{2}^{2}|\vx(t)]].
$$
Standard computations lead to the following series of equalities:
\begin{align}
     &\nonumber \mathbb{E}[\|\boldsymbol{1}-\boldsymbol{y}(t)\|_{2}^{2}|\vx(t)]\!\\
     & \nonumber\quad=\!\|\boldsymbol{1}\!-\!\mathbb{E}[\boldsymbol{y}(t)|\vx(t)]\|_{2}^{2}\!+\!\mathbb{E}[\|\boldsymbol{y}(t)\!-\!\mathbb{E}[\boldsymbol{y}(t)|\vx(t)]\|_{2}^{2}]\\
    \nonumber &\quad=\|\boldsymbol{1}-\boldsymbol{x}(t)\|_{2}^{2}\!+\!\mathbb{E}\left[\|\boldsymbol{y}(t)-\boldsymbol{x}(t)\|_{2}^{2}|\vx(t)\right]\\
    \nonumber &\quad=\|\boldsymbol{1}\!-\!\boldsymbol{x}(t)\|_{2}^{2}+\!\sum_{v \in \mathcal{V}}\mathbb{E}[(y_{v}(t)\!-\!x_{v}(t))^{2}|\vx(t)]\\
    \nonumber &\quad=\|\boldsymbol{1}\!-\!\boldsymbol{x}(t)\|_{2}^{2}\!+\!\sum_{v \in \mathcal{V}}x_{v}(t)(1\!-\!x_{v}(t))\\
    \nonumber &\quad=\sum_{v \in V}(1\!-\!x_{v}(t))^{2}\!+\!\sum_{v \in \mathcal{V}}x_{v}(t)(1\!-\!x_{v}(t))\\
    \nonumber&\quad=\sum_{v \in V} (1-x_{v}(t)).
\end{align}
By taking the expectation of this last term, we get
$$
\mathbb{E}[\|\boldsymbol{1}\!-\!\boldsymbol{y}(t)\|_{2}^{2}]=\!\|\mathbf{1}\!-\!\mathbb{E}[\boldsymbol{x}(t)]\|_{1},
$$
that, when replaced in \eqref{eq:infinite_cost1} concludes the proof.
\subsection{Further details on \eqref{eq:TV_MPC}}\label{appendix:additional_details}
As mentioned in Section~\ref{sec:alternative_map}, the approximation performed by considering the loss in \eqref{eq:cost_approx2} allows one to recover the (finite-horizon counterpart of the) loss in \eqref{eq:infinite_cost1} if both the controlled input and the mean dynamics converge to a steady state. Indeed, if the latter condition is verified, this implies  
$$
\boldsymbol{u}^{\mathrm{c}}(t) \underset{t \rightarrow \infty}{\longrightarrow} {\boldsymbol{u}}_{\mathrm{s}}^{\mathrm{c}}, 
$$
with ${\boldsymbol{u}}_{\mathrm{s}}^{\mathrm{c}}$ being the steady-state input, and, consequently,  
$$
\bar{\boldsymbol{x}}(t) \underset{t \rightarrow \infty}{\longrightarrow} \bar{\boldsymbol{x}}_{\mathrm{s}}.
$$
In turn, by continuity of the norm, the following hold
\begin{equation*}
   \|\boldsymbol{1}-\bar{\boldsymbol{x}}(t+k)\|^2_{\boldsymbol{Q}(t+k)}\rightarrow \|\boldsymbol{1}-\bar{\boldsymbol{x}}_{\mathrm{s}}\|^2_{\boldsymbol{Q}_{\mathrm{s}}}=\|\boldsymbol{1}-\bar{\boldsymbol{x}}_{\mathrm{s}}\|_1, 
\end{equation*}
with ${Q}_{\mathrm{s},v}=(|1-\bar{x}_{\mathrm{s},v}|)^{-1}$ for all $v \in \mathcal{V}$, supporting our claim. 

\bibliographystyle{IEEEtran}
\bibliography{main_social.bib}

\begin{thebibliography}{10}
\providecommand{\url}[1]{#1}
\csname url@samestyle\endcsname
\providecommand{\newblock}{\relax}
\providecommand{\bibinfo}[2]{#2}
\providecommand{\BIBentrySTDinterwordspacing}{\spaceskip=0pt\relax}
\providecommand{\BIBentryALTinterwordstretchfactor}{4}
\providecommand{\BIBentryALTinterwordspacing}{\spaceskip=\fontdimen2\font plus
\BIBentryALTinterwordstretchfactor\fontdimen3\font minus
  \fontdimen4\font\relax}
\providecommand{\BIBforeignlanguage}[2]{{%
\expandafter\ifx\csname l@#1\endcsname\relax
\typeout{** WARNING: IEEEtran.bst: No hyphenation pattern has been}%
\typeout{** loaded for the language `#1'. Using the pattern for}%
\typeout{** the default language instead.}%
\else
\language=\csname l@#1\endcsname
\fi
#2}}
\providecommand{\BIBdecl}{\relax}
\BIBdecl

\bibitem{friedkin1986formal}
N.~E. Friedkin, ``A formal theory of social power,'' \emph{Journal of
  Mathematical Sociology}, vol.~12, pp. 103--126, 1986.

\bibitem{friedkin1990social}
N.~E. Friedkin and E.~C. Johnsen, ``Social influence and opinions,''
  \emph{Journal of Mathematical Sociology}, vol.~15, no. 3--4, pp. 193--206,
  1990.

\bibitem{breschi2022driving}
V.~Breschi, C.~Ravazzi, S.~Strada, F.~Dabbene, and M.~Tanelli, ``Driving
  electric vehicles’ mass adoption: An architecture for the design of
  human-centric policies to meet climate and societal goals,''
  \emph{Transportation Research Part A: Policy and Practice}, vol. 171, p.
  103651, 2022.

\bibitem{breschi2022fostering}
------, ``Fostering the mass adoption of electric vehicles: A network-based
  approach,'' \emph{IEEE Transactions on Control of Network Systems}, vol.~9,
  no.~4, pp. 1666--1678, 2022.

\bibitem{rossi2022closed}
W.~S. Rossi, J.~W. Polderman, and P.~Frasca, ``The closed loop between opinion
  formation and personalized recommendations,'' \emph{IEEE Transactions on
  Control of Networked Systems}, vol.~9, no.~3, pp. 1092--1103, 2022.

\bibitem{castro2018opinion}
J.~Castro, J.~Lu, G.~Zhang, Y.~Dong, and L.~Mart{\'i}nez, ``Opinion
  dynamics-based group recommender systems,'' \emph{IEEE Transactions on
  Systems, Man, and Cybernetics: Systems}, vol.~48, no.~12, pp. 2394--2406,
  2018.

\bibitem{goyal2019maintaining}
M.~Goyal, D.~Chatterjee, N.~Karamchandani, and D.~Manjunath, ``Maintaining
  ferment,'' in \emph{Proc. of the 2019 IEEE 58th Conference on Decision and
  Control (CDC)}, 2019, pp. 5217--5222.

\bibitem{sprenger2024control}
B.~Sprenger, G.~De~Pasquale, R.~Soloperto, J.~Lygeros, and F.~D{\"o}rfler,
  ``Control strategies for recommendation systems in social networks,''
  \emph{arXiv preprint arXiv:2403.06152}, 2024.

\bibitem{frasca2013gossips}
P.~Frasca, C.~Ravazzi, R.~Tempo, and H.~Ishii, ``Gossips and prejudices:
  Ergodic randomized dynamics in social networks,'' \emph{IFAC Proceedings
  Volumes}, vol.~46, no.~27, pp. 212--219, 2013.

\bibitem{ravazzi2015ergodic}
C.~Ravazzi, P.~Frasca, R.~Tempo, and H.~Ishii, ``Ergodic randomized algorithms
  and dynamics over networks,'' \emph{IEEE Transactions on Control of Network
  Systems}, vol.~2, no.~1, pp. 78--87, 2015.

\bibitem{POLAK1993}
E.~Polak and T.~H. Yang, ``Moving horizon control of linear systems with input
  saturation and plant uncertainty part 1. robustness,'' \emph{International
  Journal of Control}, vol.~58, no.~3, pp. 613--638, 1993.

\end{thebibliography}

\end{document}